\documentclass[12pt,abstract]{scrartcl}
\usepackage{graphicx}
\usepackage{amsmath,amsthm}
\makeatletter
\renewcommand{\theenumii}{\@roman\c@enumii}
\makeatother

\usepackage{xcolor}
\usepackage{natbib}
\newtheorem{lemma}{Lemma}
\newtheorem{theorem}{Theorem}

\usepackage{newtxtext,newtxmath}
\begin{document}

\vspace{-.4cm}

\vspace{-.55cm}

\title{Nash equilibrium of partially asymmetric three-players zero-sum game with two strategic variables}

\author{%
Atsuhiro Satoh\thanks{atsatoh@hgu.jp}\\[.01cm]
Faculty of Economics, Hokkai-Gakuen University,\\[.02cm]
Toyohira-ku, Sapporo, Hokkaido, 062-8605, Japan,\\[.01cm]
\textrm{and} \\[.1cm]
Yasuhito Tanaka\thanks{yasuhito@mail.doshisha.ac.jp}\\[.01cm]
Faculty of Economics, Doshisha University,\\
Kamigyo-ku, Kyoto, 602-8580, Japan.\\}

\vspace{-.65cm}

\date{}

\maketitle
\thispagestyle{empty}

\vspace{-1.65cm}


\begin{abstract}
We consider a partially asymmetric three-players zero-sum game with two strategic variables. Two players (A and B) have the same payoff functions, and Player C does not. Two strategic variables are $t_i$'s and $s_i$'s for $i=A, B, C$. Mainly we will show the following results.
\begin{enumerate}
	\item The equilibrium when all players choose $t_i$'s is equivalent to the equilibrium when Players A and B choose $t_i$'s and Player C chooses $s_C$ as their strategic variables.
	\item The equilibrium when all players choose $s_i$'s is equivalent to the equilibrium when Players A and B choose $s_i$'s and Player C chooses $t_C$ as their strategic variables.
\end{enumerate}
The equilibrium when all players choose $t_i$'s and the equilibrium when all players choose $s_i$'s are not equivalent although they are equivalent in a symmetric game in which all players have the same payoff functions.
\end{abstract}

\vspace{-.4cm}

\begin{description}
\item[Keywords:] partially asymmetric three-players zero-sum game, Nash equilibrium, two strategic variables
\end{description}


\clearpage

\section{Introduction}

We consider a three-players zero-sum game with two strategic variables. Three players are Players A, B and C. Two strategic variables are $t_i$ and $s_i$, $i=A, B, C$. They are related by invertible functions. The game is symmetric for Players A and B in the sense that they have the same payoff functions. On the other hand, Player C may have a different payoff function. Thus, the game is \emph{partially asymmetric}. In Section 3 we will show the following main results. 
\begin{enumerate}
	\item The equilibrium when all players choose $t_i$'s is equivalent to the equilibrium when Players A and B choose $t_i$'s and Player C chooses $s_C$ as their strategic variables.
	\item The equilibrium when all players choose $s_i$'s is equivalent to the equilibrium when Players A and B choose $s_i$'s and Player C chooses $t_C$ as their strategic variables.
\end{enumerate}

An example of three-players zero-sum game with two strategic variables is a relative profit maximization game in a three firms oligopoly with differentiated goods. See Section \ref{ex}. In that section we will show;
\begin{enumerate}
	\item The equilibrium when all players choose $t_i$'s is not equivalent to the equilibrium when Players A and C choose $t_i$'s and Player B chooses $s_B$ as their strategic variables.
	\item The equilibrium when all players choose $t_i$'s is not equivalent to the equilibrium when Players A and B choose $s_i$'s and Player C chooses $t_C$ as their strategic variables.
	\item The equilibrium when all players choose $s_i$'s is not equivalent to the equilibrium when Players A and B choose $t_i$'s and Player C chooses $s_C$ as their strategic variables.
	\item The equilibrium when all players choose $s_i$'s is not equivalent to the equilibrium when Players A and C choose $s_i$'s and Player B chooses $t_B$ as their strategic variables.
	\item The equilibrium when all players choose $t_i$'s is not equivalent to the equilibrium when all players $s_i$'s.
\end{enumerate}

In a symmetric game,  in which all players have the same payoff functions, they are all equivalent\footnote{\cite{hst}.}.

In the next section we present a model of this paper and prove a preliminary result which is a variation of Sion's minimax theorem. 

\section{The model}

We consider a three-players zero-sum game with two strategic variables. Three players are Players A, B and C. Two strategic variables are $t_i$ and $s_i$, $i=A, B, C$. The game is symmetric for Players A and B in the sense that they have the same payoff functions. On the other hand, Player C may have a different payoff function.

$t_i$ is chosen from $T_i$ and $s_i$ is chosen from $S_i$. $T_i$ and $S_i$ are convex and compact sets in linear topological spaces, respectively, for each $i\in \{A, B, C\}$. The relations of the strategic variables are represented by 
\begin{equation*}
s_i=f_i(t_A, t_B, t_C),\ i=A, B, C,
\end{equation*}
and 
\begin{equation*}
t_i=g_i(s_A, s_B, s_C),\ i=A, B, C.
\end{equation*}
$(f_A, f_B, f_C)$ and $(g_A, g_B, g_C)$ are continuous, invertible, one-to-one and onto functions. When Players A and B choose $t_A$ and $t_B$ and Player C chooses $s_C$, then $t_C$ is determined according to 
\begin{equation*}
t_C=g_C(f_A(t_A, t_B, t_C),f_B(t_A, t_B, t_C), s_C).
\end{equation*}
We denote this $t_C$ by $t_C(t_A, t_B, s_C)$. 

When Players A and B choose $s_A$ and $s_B$ and Player C chooses $t_C$, then $t_A$ and $t_B$ are determined according to 
\begin{equation*}
\left\{ 
\begin{array}{l}
t_A=g_A(s_A, s_B, f_C(t_A, t_B, t_C))\\
t_B=g_B(s_A, s_B, f_C(t_A, t_B, t_C)). 
\end{array}%
\right.
\end{equation*}
We denote these $t_A$ and $t_B$ by $t_A(s_A, s_B, t_C)$ and $t_B(s_A, s_B,
t_C)$.

When all players choose $s_{A}$, $s_{B}$ and $s_{C}$, $t_{A}$, then $t_{B}$ and $t_{C}$ are determined according to 
\begin{equation*}
t_{A}=g_{A}(s_{A},s_{B},s_{C}),\ t_{B}=g_{B}(s_{A},s_{B},s_{C}),\
t_{C}=g_{C}(s_{A},s_{B},s_{C}).
\end{equation*}%
Denote these $t_{A}$, $t_{B}$ and $t_{C}$ by $t_{A}(s_{A},s_{B},s_{C})$, $%
t_{B}(s_{A},s_{B},s_{C})$ and $t_{C}(s_{A},s_{B},s_{C})$.

The payoff function of Player $i$ is $u_i,\ i=A, B, C$. It is written as 
\begin{equation*}
u_i(t_A, t_B, t_C),\ i\in \{A, B, C\}.
\end{equation*}
We assume

\begin{quote}
$u_i: T_1\times T_2\times T_3\Rightarrow \mathbb{R}$ for each $i\in \{A, B, C\}$ is continuous on $T_1\times T_2\times T_3$. Thus, it is continuous on $S_1\times S_2\times S_3$ through $f_i,\ i=A, B, C$. It is quasi-concave on $T_i$ and $S_i$ for a strategy of each other player, and quasi-convex on $T_j,\ j\neq i$ and $S_j,\ j\neq i$ for each $t_i$ and $s_i$.
\end{quote}

We do not assume differentiability of the payoff functions.

Symmetry of the game for Players A and B means that in the payoff function of each player, Players A and B are interchangeable. Since the game is a zero-sum game, the sum of the values of the payoff functions of the players is zero.

We assume that all $T_i$'s are identical, and all $S_i$'s are identical. Denote them by $T$ and $S$.

Sion's minimax theorem (\cite{sion}, \cite{komiya}, \cite{kind}) for a
continuous function is stated as follows.

\begin{lemma}
Let $X$ and $Y$ be non-void convex and compact subsets of two linear topological spaces, and let $f:X\times Y \rightarrow \mathbb{R}$ be a function that is continuous and quasi-concave in the first variable and continuous and quasi-convex in the second variable, then
\begin{equation*}
\max_{x\in X}\min_{y\in Y}f(x,y)=\min_{y\in Y}\max_{x\in X}f(x,y).
\end{equation*}
\label{l1}
\end{lemma}

We follow the description of Sion's theorem in \cite{kind}. 

Applying this lemma to the situation of this paper, we have the following relations.
\begin{equation*}
\max_{t_{A}\in T}\min_{t_{B}\in T}u_{A}(t_{A},t_{B},t_{C})=\min_{t_{B}\in
T}\max_{t_{A}\in T}u_{A}(t_{A},t_{B},t_{C}),\ \max_{t_{B}\in
T}\min_{t_{A}\in T}u_{B}(t_{A},t_{B},t_{C})=\min_{t_{A}\in T}\max_{t_{B}\in
T}u_{B}(t_{A},t_{B},t_{C}).
\end{equation*}%
\begin{equation*}
\max_{t_{A}\in T}\min_{t_{C}\in T}u_{A}(t_{A},t_{B},t_{C})=\min_{t_{C}\in
T}\max_{t_{A}\in T}u_{A}(t_{A},t_{B},t_{C}),\ \max_{t_{B}\in
T}\min_{t_{C}\in T}u_{B}(t_{A},t_{B},t_{C})=\min_{t_{C}\in T}\max_{t_{B}\in
T}u_{B}(t_{A},t_{B},t_{C}).
\end{equation*}%
\begin{align*}
& \max_{t_{A}\in T}\min_{t_{B}\in
T}u_{A}(t_{A},t_{B},t_{C}(t_{A},t_{B},s_{C}))=\min_{t_{B}\in
T}\max_{t_{A}\in T}u_{A}(t_{A},t_{B},t_{C}(t_{A},t_{B},s_{C})), \\
& \max_{t_{B}\in T}\min_{t_{A}\in
T}u_{B}(t_{A},t_{B},t_{C}(t_{A},t_{B},s_{C}))=\min_{t_{A}\in
T}\max_{t_{B}\in T}u_{B}(t_{A},t_{B},t_{C}(t_{A},t_{B},s_{C})).
\end{align*}%
\begin{align*}
& \max_{t_{A}\in T}\min_{s_{C}\in
S}u_{A}(t_{A},t_{B},t_{C}(t_{A},t_{B},s_{C}))=\min_{s_{C}\in
S}\max_{t_{A}\in T}u_{A}(t_{A},t_{B},t_{C}(t_{A},t_{B},s_{C})),\\
& \max_{t_{B}\in T}\min_{s_{C}\in
S}u_{B}(t_{A},t_{B},t_{C}(t_{A},t_{B},s_{C}))=\min_{s_{C}\in
S}\max_{t_{B}\in T}u_{B}(t_{A},t_{B},t_{C}(t_{A},t_{B},s_{C})), \\
\end{align*}%

Further we show the following result.

\begin{lemma}
\begin{align*}
&\min_{t_C\in T}\max_{t_A\in T}u_A(t_A, t_B, t_C)=\min_{s_C\in
S}\max_{t_A\in T}u_A(t_A, t_B, t_C(t_A,t_B,s_C)) \\
&=\max_{t_A\in T}\min_{s_C\in S}u_A(t_A, t_B, t_C(t_A,t_B,s_C))=\max_{t_A\in
T}\min_{t_C\in T}u_A(t_A, t_B, t_C),
\end{align*}
and 
\begin{align*}
&\min_{t_C\in T}\max_{t_B\in T}u_B(t_A, t_B, t_C)=\min_{s_C\in
S}\max_{t_B\in T}u_B(t_A, t_B, t_C(t_A,t_B,s_C)) \\
&=\max_{t_B\in T}\min_{s_C\in S}u_B(t_A, t_B, t_C(t_A,t_B,s_C))=\max_{t_B\in
T}\min_{t_C\in T}u_B(t_A, t_B, t_C).
\end{align*}
\label{l3}
\end{lemma}

\begin{proof}
$\max_{t_A\in T}u_A(t_A, t_B, t_C(t_A,t_B,s_C))$ is the maximum of $u_A$ with respect to $t_A$ given $t_B$ and $s_C$. Let $\tilde{t}_A(s_C)=\arg\max_{t_A\in T}u_A(t_A, t_B, t_C(t_A,t_B,s_C))$, and fix the value of $t_C$ at
\begin{equation}
t_C^0=g_C(f_A(\tilde{t}_A(s_C), t_B, t_C^0), f_B(\tilde{t}_A(s_C), t_B, t_C^0), s_C).\label{tc1}
\end{equation}
We have
\begin{align*}
\max_{t_A\in T}u_A(t_A, t_B, t_C^0)\geq u_A(\tilde{t}_A(s_C), t_B, t_C^0)=\max_{t_A\in T}u_A(t_A, t_B, t_C(t_A,t_B,s_C)),
\end{align*}
where $\max_{t_A\in T}u_A(t_A, t_B, t_C^0)$ is the maximum of $u_A$ with respect to $t_A$ given the value of $t_C$ at $t_C^0$. We assume that $\tilde{t}_A(s_C)=\arg\max_{t_A\in T}u_A(t_A, t_B, t_C(t_A,t_B,s_C))$ is single-valued. By the maximum theorem and continuity of $u_A$, $\tilde{t}_A(s_C)$ is continuous, then any value of $t_C^0$ can be realized by appropriately choosing $s_C$ given $t_B$ according to (\ref{tc1}). Therefore,
\begin{equation}
\min_{t_C\in T}\max_{t_A\in T}u_A(t_A, t_B, t_C)\geq \min_{s_C\in S}\max_{t_A\in T}u_A(t_A, t_B, t_C(t_A,t_B,s_C)).\label{4-11}
\end{equation}

On the other hand, $\max_{t_A\in T}u_A(t_A, t_B, t_C)$ is the maximum of $u_A$ with respect to $t_A$ given $t_B$ and $t_C$. Let $\tilde{t}_A(t_C)=\arg\max_{t_A\in T}u_A(t_A, t_B, t_C)$, and fix the value of $s_C$ at
\begin{equation}
s_C^0=f_C(\tilde{t}_A(t_C), t_B, t_C).\label{sc1}
\end{equation}
Thus, we have
\begin{align*}
\max_{t_A\in T}u_A(t_A, t_B, t_C(t_A,t_B,s_C^0))\geq u_A(\tilde{t}_A(s_C), t_B, t_C(t_A,t_B,s_C^0))=\max_{t_A\in T}u_A(t_A, t_B, t_C),
\end{align*}
where $\max_{t_A\in T}u_A(t_A, t_B, t_C(t_A,t_B,s_C^0))$ is the maximum of $u_A$ with respect to $t_A$ given the value of $s_C$ at $s_C^0$. We assume that $\tilde{t}_A(t_C)=\arg\max_{t_A\in T}u_A(t_A, t_B, t_C)$ is single-valued. By the maximum theorem and continuity of $u_A$, $\tilde{t}_A(t_C)$ is continuous, then any value of $s_C^0$ can be realized by appropriately choosing $t_C$ given $t_B$ according to (\ref{sc1}). Therefore,
\begin{equation}
\min_{s_C\in S}\max_{t_A\in T}u_A(t_A, t_B, t_C(t_A,t_B,s_C))\geq \min_{t_C\in S}\max_{t_A\in T}u_A(t_A, t_B, t_C).\label{4-21}
\end{equation}
Combining (\ref{4-11}) and (\ref{4-21}), we get
\[\min_{s_C\in S}\max_{t_A\in T}u_A(t_A, t_B, t_C(t_A,t_B,s_C))=\min_{t_C\in S}\max_{t_A\in T}u_A(t_A, t_B, t_C).\]
Since any value of $s_C$ can be realized by appropriately choosing $t_C$ given $t_A$ and $t_B$, we have
\[\min_{s_C\in S}u_A(t_A, t_B, t_C(t_A,t_B,s_C))=\min_{t_C\in S}u_A(t_A, t_B, t_C).\]
Thus,
\[\max_{t_A\in T}\min_{s_C\in S}u_A(t_A, t_B, t_C(t_A,t_B,s_C))=\max_{t_A\in T}\min_{t_C\in S}u_A(t_A, t_B, t_C).\]
Therefore,
\begin{align*}
&\min_{t_C\in T}\max_{t_A\in T}u_A(t_A, t_B, t_C)=\min_{s_C\in S}\max_{t_A\in T}u_A(t_A, t_B, t_C(t_A,t_B,s_C)),\\
=&\max_{t_A\in T}\min_{s_C\in S}u_A(t_A, t_B, t_C(t_A,t_B,s_C))=\max_{t_A\in T}\min_{t_C\in T}u_A(t_A, t_B, t_C),
\end{align*}
given $t_B$.

By similar procedures, we can show
\begin{align*}
&\min_{t_C\in T}\max_{t_B\in T}u_B(t_A, t_B, t_C)=\min_{s_C\in S}\max_{t_B\in T}u_B(t_A, t_B, t_C(t_A,t_B,s_C)),\\
=&\max_{t_B\in T}\min_{s_C\in S}u_B(t_A, t_B, t_C(t_A,t_B,s_C))=\max_{t_B\in T}\min_{t_C\in T}u_B(t_A, t_B, t_C),
\end{align*}
given $t_A$. 
\end{proof}

\section{The main results}

In this section we present the following main result of this paper.
\begin{theorem}
The equilibrium when all players choose $t_i$'s is equivalent to the equilibrium when Player C chooses $s_C$ and Players A and B choose $t_i$'s as their strategic variables.\label{t1}
\end{theorem}

\begin{proof}
\begin{enumerate}
	\item Consider a situation $(t_A,t_B,t_C)=(t,t,t_C)$. By symmetry for Players A and B,
\[\max_{t_A\in T}u_A(t_A, t, t_C)=\max_{t_B\in T}u_B(t,t_B,t_C),\]
and
\[\arg\max_{t_A\in T}u_A(t_A, t, t_C)=\arg\max_{t_B\in T}u_B(t,t_B,t_C)\in T,\]
given $t_C$. Let
\[t_C(t)=\arg\max_{t_C\in T}u_C(t,t,t_C).\]
We assume that it is a single-valued continuous function. 

Consider the following function.
\[t\rightarrow \arg\max_{t_A\in T}u_A(t_A, t, t_C),\ \mathrm{given}\ t_C.\]
This function is continuous and $T$ is compact. Thus, there exists a fixed point given $t_C$. Denote it by $t^*(t_C)$, then
\[t^*(t_C)=\arg\max_{t_A\in T}u_A(t_A, t^*(t_C), t_C)=\arg\max_{t_B\in T}u_B(t^*(t_C), t_B,t_C),\ \mathrm{given}\ t_C.\]
Now we consider the following function.
\[t\rightarrow t^*(t_C(t)).\]
This also has a fixed point. Denote it by $t^*$ and $t_C(t^*)$ by $t_C^*$, then we have
\[t^*=\arg\max_{t_A\in T}u_A(t_A, t^*, t_C^*)=\arg\max_{t_B\in T}u_B(t^*, t_B,t_C^*),\]
\[t_C^*=\arg\max_{t_C\in T}u_C(t^*,t^*,t_C).\]
\[\max_{t_A\in T}u_A(t_A, t^*, t_C^*)=u_A(t^*, t^*, t_C^*)=\max_{t_B\in T}u_B(t^*, t_B, t_C^*)=u_B(t^*, t^*, t_C^*),\]
and
\[\max_{t_C\in T}u_C(t^*, t^*, t_C)=u_C(t^*, t^*, t_C^*).\]
$(t_A,t_B,t_C)=(t^*,t^*,t_C^*)$ is a Nash equilibrium when all players choose $t_i$'s

\item Because the game is zero-sum,
\[u_A(t^*, t^*, t_C)+u_B(t^*, t^*,t_C)+u_C(t^*, t^*,t_C)=0.\]
By symmetry $u_A(t^*, t^*,t_C)=u_B(t^*, t^*,t_C)$. Thus,
\[2u_A(t^*, t^*,t_C)+u_C(t^*, t^*,t_C)=0.\]
This means
\[2u_A(t^*, t^*,t_C)=-u_C(t^*, t^*,t_C),\]
and
\[2\min_{t_C\in T}u_A(t^*, t^*,t_C)=-\max_{t_C\in T}u_C(t^*, t^*,t_C).\]
From this and symmetry for Players A and B, we get
\[\arg\min_{t_C\in T}u_A(t^*, t^*,t_C)=\arg\min_{t_C\in T}u_B(t^*, t^*,t_C)=\arg\max_{t_C\in T}u_C(t^*, t^*,t_C)=t_C^*.\]
We have
\[\min_{t_C\in T}u_A(t^*, t^*, t_C^*)=u_A(t^*, t^*, t_C^*)=\max_{t_A\in T}u_A(t_A,t^*,t_C^*),\]
\[\min_{t_C\in T}u_B(t^*, t^*, t_C^*)=u_B(t^*, t^*, t_C^*)=\max_{t_B\in T}u_B(t^*,t_B,t_C^*).\]
Therefore,
\[\min_{t_C\in T}\max_{t_A\in T}u_A(t_A, t^*, t_C)\leq \max_{t_A\in T}u_A(t_A, t^*, t_C^*)=\min_{t_C\in T}u_A(t^*,t^*, t_C)\leq \max_{t_A\in T}\min_{t_C\in T}u_A(t_A,t^*, t_C).\]
From Lemma \ref{l3} we obtain
\begin{align}
&\min_{t_C\in T}\max_{t_A\in T}u_A(t_A, t^*, t_C)=\max_{t_A\in T}u_A(t_A, t^*, t_C^*)=\min_{t_C\in T}u_A(t^*,t^*, t_C)\label{l3-1}\\
&=\max_{t_A\in T}\min_{t_C\in T}u_A(t_A,t^*, t_C)=\min_{s_C\in S}\max_{t_A\in T}u_A(t_A, t^*, t_C(t_A,t^*,s_C))\notag\\
&=\max_{t_A\in T}\min_{s_C\in T}u_A(t_A,t^*, t_C(t_A,t^*,s_C)).\notag
\end{align}

\item Let
\[s^0_C(t^*)=f_C(t^*,t^*,t_C^*).\]
Since any value of $s_C$ can be realized by appropriately choosing $t_C$,
\begin{equation}
\min_{s_C\in S}u_A(t^*,t^*,t_C(t^*,t^*,s_C))=\min_{t_C\in T}u_A(t^*,t^*,t_C)=u_A(t^*,t^*,t_C^*).\label{z1}
\end{equation}
Thus,
\[\arg\min_{s_C\in S}u_A(t^*,t^*,t_C(t^*,t^*,s_C))=s^0_C(t^*).\]
(\ref{l3-1}) and (\ref{z1}) mean
\begin{equation}
\min_{s_C\in S}\max_{t_A\in T}u_A(t_A,t^*,t_C(t_A,t^*,s_C))=\min_{s_C\in S}u_A(t^*,t^*,t_C(t^*,t^*,s_C)).\label{z2}
\end{equation}
We have
\[\max_{t_A\in T}u_A(t_A,t^*,t_C(t_A,t^*,s_C))\geq u_A(t^*,t^*,t_C(t^*,t^*,s_C)).\]
Therefore,
\[\arg\min_{s_C\in S}\max_{t_A\in T}u_A(t_A,t^*,t_C(t_A,t^*,s_C))=\arg\min_{s_C\in S}u_A(t^*,t^*,t_C(t^*,t^*,s_C))=s^0_C(t^*)\]
Thus, by (\ref{z2})
\begin{align*}
&\min_{s_C\in S}\max_{t_A\in T}u_A(t_A,t^*,t_C(t_A,t^*,s_C))=\max_{t_A\in T}u_A(t_A,t^*,t_C(t^*,t^*,s^0_C(t^*)))\\
=&\min_{s_C\in S}u_A(t^*,t^*,t_C(t^*,t^*,s_C))=u_A(t^*,t^*,t_C(t^*,t^*,s^0_C(t^*))).
\end{align*}
Therefore,
\begin{equation}
\arg\max_{t_A\in T}u_A(t_A,t^*,t_C(t_A,t^*,s^0_C(t^*)))=t^*.\label{t1-2}
\end{equation}
By symmetry for Players A and B,
\begin{equation}
\arg\max_{t_B\in T}u_B(t^*,t_B,t_C(t^*,t_B,s^0_C(t^*)))=t^*.\label{t1-3}
\end{equation}

On the other hand, because any value of $s_C$ is realized by appropriately choosing $t_C$,
\[\max_{s_C\in S}u_C(t^*,t^*,t_C(t^*,t^*,s_C))=\max_{t_C\in T}u_C(t^*,t^*,t_C)=u_C(t^*,t^*,t_C^*).\]
Therefore,
\begin{equation}
\arg\max_{s_C\in S}u_C(t^*,t^*,t_C(t^*,t^*,s_C))=s^0_C(t^*)=f_C(t^*,t^*,t_C^*).\label{t1-1}
\end{equation}

From (\ref{t1-2}), (\ref{t1-3}) and (\ref{t1-1}), $(t^*,t^*,t_C(t^*,t^*,s^0_C(t^*)))$ is a Nash equilibrium which is equivalent to $(t^*,t^*,t_C^*)$. 
\end{enumerate}
\end{proof}

Interchanging $t_i$ and $s_i$ for each player, we can show 
\begin{theorem}
The equilibrium when all players choose $s_i$'s is equivalent to the equilibrium when Player C chooses $t_C$ and Players A and B choose $s_i$'s as their strategic variables.\label{t2}
\end{theorem}

\section{Various examples}\label{ex}

Consider a game of relative profit maximization under oligopoly including three firms with differentiated goods\footnote{%
About relative profit maximization in an oligopoly see \cite{mm}, \cite{ebl2}, \cite{eb2}, \cite{st}, \cite{eb1}, \cite{ebl1} and  \cite{redondo}}. It is a three-players zero-sum game with two strategic variables. The firms are A, B and C. The strategic variables are the outputs and the prices of their goods. We consider the following six patterns of competition. 

\begin{enumerate}
\item Pattern 1: All firms determine their outputs. It is a Cournot case.

The inverse demand functions are 
\begin{equation*}
p_A=a-x_A-bx_B-bx_C,
\end{equation*}
\begin{equation*}
p_B=a-x_B-bx_A-bx_C,
\end{equation*}
and 
\begin{equation*}
p_C=a-x_C-bx_A-bx_B,
\end{equation*}
where $0<b<1$. $p_A$, $p_B$ and $p_C$ are the prices of the goods of Firms A, B and C, and $x_A$, $x_B$ and $x_C$ are the outputs of them.

\item Pattern 2: Firms A and B determine their outputs, and Firm C determines the price of its good.

From the inverse demand functions, 
\begin{equation*}
p_A=(1-b)a+b^2x_B-bx_B+b^2x_A-x_A+bp_C,
\end{equation*}
\begin{equation*}
p_B=(1-b)a+b^2x_B-x_B+b^2x_A-bx_A+bp_C,
\end{equation*}
and 
\begin{equation*}
x_C=a-bx_B-bx_A-p_C
\end{equation*}
are derived.

\item Pattern 3: Firms A and C determine their outputs, and Firm B determines the price of its good.

From the inverse demand functions, 
\begin{equation*}
p_A=(1-b)a+b^2x_C-bx_C+b^2x_A-x_A+bp_B,
\end{equation*}
\begin{equation*}
p_C=(1-b)a+b^2x_C-x_C+b^2x_A-bx_A+bp_B,
\end{equation*}
and 
\begin{equation*}
x_B=a-bx_C-bx_A-p_B
\end{equation*}
are derived.

\item Pattern 4: Firms A and B determine the prices of their goods, and Firm C determines its output.

From the above inverse demand functions, we obtain 
\begin{equation*}
p_C=\frac{(1-b)a+2b^2x_C-bx_C-x_C+bp_A+bp_B}{1+b},
\end{equation*}
\begin{equation*}
x_B=\frac{(1-b)a+b^2x_C-bx_C+bp_A-p_B}{(1-b)(1+b)},
\end{equation*}
and 
\begin{equation*}
x_A=\frac{(1-b)a+b^2x_C-bx_C-p_A+bp_B}{(1-b)(1+b}.
\end{equation*}

\item Pattern 5: Firms A and C determine the prices of their goods, and Firm B determines its output.

From the above inverse demand functions, we obtain 
\begin{equation*}
p_B=\frac{(1-b)a+2b^2x_B-bx_B-x_B+bp_C+bp_A}{1+b},
\end{equation*}
\begin{equation*}
x_A=\frac{(1-b)a+b^2x_B-bx_B+bp_C-p_A}{(1-b)(1+b)},
\end{equation*}
and 
\begin{equation*}
x_C=\frac{(1-b)a+b^2x_B-bx_B-p_C+bp_A}{(1-b)(1+b}.
\end{equation*}

\item Pattern 6: All firms determine the prices of their goods. It is a Bertrand case.

From the inverse demand functions, the direct demand functions are derived as
follows; 
\begin{equation*}
x_A=\frac{(1-b)a-(1+b)p_A+b(p_A+p_C)}{(1-b)(1+2b)},
\end{equation*}
\begin{equation*}
x_B=\frac{(1-b)a-(1+b)p_B+b(p_B+p_C)}{(1-b)(1+2b)},
\end{equation*}
and 
\begin{equation*}
x_C=\frac{(1-b)a-(1+b)p_C+b(p_A+p_B)}{(1-b)(1+2b)}.
\end{equation*}
\end{enumerate}

The absolute profits of the firms are 
\begin{equation*}
\pi_A=p_Ax_A-c_Ax_A,
\end{equation*}
\begin{equation*}
\pi_B=p_Bx_B-c_Bx_B,
\end{equation*}
and 
\begin{equation*}
\pi_C=p_Cx_C-c_Cx_C.
\end{equation*}
$c_A$, $c_B$ and $c_C$ are the constant marginal costs of Firms A, B and C. The relative profits of the firms are 
\begin{equation*}
\psi_A=\pi_A-\frac{\pi_B+\pi_C}{2},
\end{equation*}
\begin{equation*}
\psi_B=\pi_B-\frac{\pi_A+\pi_C}{2},
\end{equation*}
and 
\begin{equation*}
\psi_C=\pi_C-\frac{\pi_A+\pi_B}{2}.
\end{equation*}

The firms determine the values of their strategic variables to maximize the
relative profits. We see 
\begin{equation*}
\psi_A+\psi_B+\psi_C=0,
\end{equation*}
so the game is zero-sum. We assume $c_A=c_B$, that is, the game is symmetric for Firms A and B. However, $c_C$ is not equal to $c_A$. Thus, the game is partially asymmetric.

We calculate the equilibrium outputs of the firms in the above six patterns.
\begin{enumerate}
	\item Pattern 1
\[x_A=\frac{bc_C-4c_A-ab+4a}{(4-b)(b+2)},\]
\[x_B=\frac{bc_C-4c_A-ab+4a}{(4-b)(b+2)},\]
\[x_C=\frac{bc_C+4c_C-2bc_A+ab-4a}{(4-b)(b+2)}.\]

	\item Pattern 2
\[x_A=\frac{bc_C-4c_A-ab+4a}{(4-b)(b+2)},\]
\[x_B=\frac{bc_C-4c_A-ab+4a}{(4-b)(b+2)},\]
\[x_C=\frac{bc_C+4c_C-2bc_A+ab-4a}{(b-4)(b+2)}.\]

	\item Pattern 3
\[x_A=\frac{5b^2c_C+4bc_C-3b^3c_A+6b^2c_A+4bc_A-16c_A+3ab^3-11ab^2-8ab+16a}{(4-b)(1-b)(b+2)(3b+4)},\]
\[x_B=\frac{bc_C-4c_A-ab+4a}{(4-b)(b+2)},\]
\[x_C=\frac{7b^2c_C-16c_C-3b^3c_A+4b^2c_A+8bc_A+3ab^3-11ab^2-8ab+16a}{(4-b)(1-b)(b+2)(3b+4)}.\]

	\item Pattern 4
\[x_A=\frac{2b^2c_C+bc_C+3b^2c_A-2bc_A-4c_A-5ab^2+ab+4a}{(1-b)(b+2)(5b+4))},\]
\[x_B=\frac{2b^2c_C+bc_C+3b^2c_A-2bc_A-4c_A-5ab^2+ab+4a}{(1-b)(b+2)(5b+4))},\]
\[x_C=\frac{b^2c_C-3bc_C-4c_C+4b^2c_A+2bc_A-5ab^2+ab+4a}{(1-b)(b+2)(5b+4))}.\]

	\item Pattern 5
\[x_A=\frac{3b^2c_C-b^3c_C+4bc_C+6b^3c_A+16b^2c_A-12bc_A-16c_A-5ab^3-19ab^2+8ab+16a}{(1-b)(b+2)(b+4)(5b+4)},\]
\[x_B=\frac{2b^2c_C+bc_C+3b^2c_A-2bc_A-4c_A-5ab^2+ab+4a}{(1-b)(b+2)(5b+4)},\]
\[x_C=\frac{4b^3c_C+7b^2c_C-16bc_C-16c_C+b^3c_A+12b^2c_A+8bc_A-5ab^3-19ab^2+8ab+16a}{(1-b)(b+2)(b+4)(5b+4)}.\]

	\item Pattern 6
\[x_A=\frac{2b^2c_C+bc_C+3b^2c_A-2bc_A-4c_A-5ab^2+ab+4a}{(1-b)(b+2)(5b+4)},\]
\[x_B=\frac{2b^2c_C+bc_C+3b^2c_A-2bc_A-4c_A-5ab^2+ab+4a}{(1-b)(b+2)(5b+4)},\]
\[x_C=\frac{b^2c_C-3bc_C-4c_C+4b^2c_A+2bc_A-5ab^2+ab+4a}{(1-b)(b+2)(5b+4)}.\]
\end{enumerate}
We find that Pattern 1 is equivalent to Pattern 2 (an example of Theorem \ref{t1}), but it is not equivalent to Pattern 3, and that Pattern 6 is equivalent to Pattern 4 (an example of Theorem \ref{t2}), but it is not equivalent to Pattern 5. Pattern 1 (Cournot Pattern) and Pattern 6 (Bertrand) are not equivalent unless we have $c_C=c_A$.

\section{Concluding Remarks}

In this paper we have examined equilibria in a partially asymmetric three-players zero-sum game under various situations. We want to extend the results of this paper to a general multi-players zero-sum game.

\section*{Acknowledgment}

This work was supported by Japan Society for the Promotion of Science KAKENHI Grant Number 15K03481  and 18K01594.

\end{document}